\documentclass[runningheads]{llncs}

\usepackage{cite} 

\usepackage{multirow}

\usepackage{adjustbox}
\usepackage{amsmath}
\usepackage{amssymb}
\usepackage{framed}
\usepackage{booktabs} 
\usepackage[table]{xcolor}

\usepackage{standalone}
\usepackage{pgfplots}
\usepgfplotslibrary{groupplots}

\usepackage{xspace}
\usepackage{placeins} 
\usepackage{relsize}
\usepackage[inline]{enumitem} 
  \newenvironment{inumerate}{\begin{enumerate*}[label=(\roman*),itemjoin=\hspace{3pt}]}{\end{enumerate*}}
	
\usepackage{csquotes} 
\usepackage{enumitem}
\usepackage{mathtools}
\usepackage[draft]{todonotes}
\usepackage{booktabs}
\usepackage{graphicx}
\usepackage{pifont}
\usepackage{siunitx}

\usepackage{array} 
\newcolumntype{C}[1]{>{\centering\arraybackslash\hspace{0pt}}p{#1}}
\newcolumntype{R}[1]{>{\raggedleft\arraybackslash\hspace{0pt}}p{#1}}

\PassOptionsToPackage{hyphens}{url}
\usepackage[
    pdftex,
    pdftitle={Topology-Hiding Path Validation for Large-Scale Quantum Key Distribution Networks},
    pdfauthor={Stephan Krenn, Thomas Lorünser, Omid Mir, Sebastian Ramacher, Florian Wohner},
pdfpagelabels=true, 
    bookmarks,bookmarksopen,bookmarksdepth=3,
    breaklinks,colorlinks,citecolor=blue,linkcolor=black,urlcolor=blue
   ]{hyperref}

\usepackage[capitalize,nameinlink,sort]{cleveref}

\usepackage{float}
\floatstyle{boxed}
\newfloat{construction}{ht}{lo}
\floatname{construction}{Construction}
\crefname{construction}{Construction}{Constructions}

\usepackage{etoolbox}
\makeatletter
\patchcmd{\thebibliography}{\section{\refname}}{\section*{\refname}}{}{}
\patchcmd{\thebibliography}{\section{\bibname}}{\section*{\bibname}}{}{}
\makeatother

 \graphicspath{{figures/}}

\newcommand{\comment}[1]{\raisebox{.25ex}{{\color{gray} \text{\smaller\smaller\tt//#1}}}}

\newcommand{\negl}{\mathsf{negl}}

\newcommand{\node}{\mathtt{N}}
\newcommand{\issuer}{\mathtt{I}}
\newcommand{\sid}{\mathtt{sid}}
\newcommand{\did}{\mathtt{did}}

\newcommand{\nym}{\mathsf{nym}}
\newcommand{\pk}{\mathsf{pk}}
\newcommand{\sk}{\mathsf{sk}}
\newcommand{\msg}{\mathsf{msg}}
\newcommand{\state}{\mathsf{st}}
\newcommand{\cred}{\mathsf{cred}}
\newcommand{\scope}{\mathsf{scope}}
\newcommand{\ctx}{\mathsf{ctx}}
\newcommand{\attr}{a}
\newcommand{\predicate}{\phi}

\newcommand{\adv}{\mathcal{A}}
\newcommand{\oracle}{\mathcal{O}}
\newcommand{\setS}{\mathcal{S}}
\newcommand{\network}{\mathbf{G}}

\newcommand{\getsr}{\gets_{\$}}

\newcommand{\zo}{\{0,1\}}

\newcommand{\GG}{\mathbb{G}}
\newcommand{\ZZ}{\mathbb{Z}}

\newcommand{\secpar}{\lambda}
\newcommand{\pp}{pp}

\newcommand{\Sig}{\Sigma}
\newcommand{\Nym}{\Psi}

\newcommand{\alg}[1]{\mathsf{#1}}
\newcommand{\NIZK}{\alg{NIZK}}
\newcommand{\hash}{\alg{H}}

\newcommand{\NymGen}{\alg{NymGen}}
\newcommand{\Register}{\alg{Register}}
\newcommand{\Transfer}{\alg{Send}}
\newcommand{\Alice}{\alg{A}}
\newcommand{\Bob}{\alg{B}}

\newcommand{\Groth}{\alg{Groth}}
\newcommand{\ParGen}{\alg{ParGen}}
\newcommand{\KeyGen}{\alg{KeyGen}}
\newcommand{\Sign}{\alg{Sign}}
\newcommand{\Rand}{\alg{Rand}}
\newcommand{\Verify}{\alg{Verify}}
\newcommand{\SENym}{\alg{Nym}}

\makeatletter
\AtBeginDocument{%
	\def\doi#1{\url{https://doi.org/#1}}}
\makeatother

\begin{document}

\title{Topology-Hiding Path Validation for Large-Scale Quantum Key Distribution Networks}
\titlerunning{Topology-Hiding Path Validation for Large-Scale QKD Networks}

\author{Stephan Krenn\inst{1} \and 
	    Omid Mir\inst{1} \and
	    Thomas Lorünser\inst{1}\textsuperscript{,}\inst{2} \and \\
	    Sebastian Ramacher\inst{1} \and 
	    Florian Wohner\inst{1}}
\authorrunning{S. Krenn et al.}
\institute{AIT Austrian Institute of Technology, Vienna, Austria \\
           \email{\href{mailto:stephan.krenn@ait.ac.at,omid.mir@ait.ac.at,thomas.loruenser@ait.ac.at,sebastian.ramacher@ait.ac.at,florian.wohner@ait.ac.at?subject=Question regarding your paper on topology validation in QKD networks }{\{firstname.lastname\}@ait.ac.at}}
\and
Digital Factory Vorarlberg GmbH, Dornbirn, Austria
}

\maketitle              

\begingroup
\renewcommand\thefootnote{}
\footnotetext{This is the author-accepted manuscript of a paper accepted for publication in the proceedings of the 24th International Conference on Applied Cryptography and Network Security (ACNS 2026), to appear in Springer LNCS.}
\endgroup

\begin{abstract}
  Secure long-distance communication in quantum key distribution (QKD) networks depends on trusted repeater nodes along the entire transmission path. Consequently, these nodes will be subject to strict auditing and certification in future large-scale QKD deployments.
  However, trust must also extend to the network operator, who is responsible for fulfilling contractual obligations--such as ensuring certified devices are used and transmission paths remain disjoint where required.
  
  In this work, we present a path validation protocol specifically designed for QKD networks. It enables the receiver to verify compliance with agreed-upon policies. At the same time, the protocol preserves the operator’s confidentiality by ensuring that no sensitive information about the network topology is revealed to users.
  
  We provide a formal model and a provably secure generic construction of the protocol, along with a concrete instantiation. For long-distance communication involving $100$ nodes, the protocol has a computational cost of $1-2.5$\,s depending on the machine, and a communication overhead of less than $70$\,kB--demonstrating the efficiency of our approach.

\end{abstract} 

\section{Introduction}\label{sec:intro}
  Quantum Key Distribution (QKD)~\cite{C:BenBra84} enables secure key exchange by utilizing quantum mechanics, specifically superposition and entanglement, to ensure that any eavesdropping attempt can be detected. 
  Unlike classical cryptographic methods, which rely on the difficulty of certain mathematical problems, QKD is based on the laws of physics and achieves unconditional security, making it resistant against classical and quantum attacks.
  QKD is thus expected to play a crucial role in highly security-sensitive application domains in the future, and its practicability has been demonstrated in numerous testbeds, e.g.,~\cite{DBLP:journals/entropy/BrauerVBMBGRBFPPLMB24,DBLP:conf/eicc/HenrichHSS23,DBLP:conf/icfsp/GeitzDB23,DBLP:conf/ecoc/KutscheraDBLMVH21}.
  
  However, QKD faces limitations, primarily the distance over which quantum signals can travel due to photon loss in optical fibers or through free space. 
  To overcome this limitation, trusted repeaters are employed to relay quantum information over long distances~\cite{DBLP:journals/csur/MehicNRMPAMSPPV20,DBLP:journals/corr/abs-2210-01636}, typically placed at most every  $100-200$\,km, leading to potentially tens of repeaters in a terrestrial long-distance connection. 
  These repeaters allow the key exchange to continue securely, but they also introduce vulnerabilities, as a compromised repeater could potentially intercept or alter the key, thereby undermining the security guarantees of the entire network.
  Consequently, minimizing the necessary trust required for trusted repeaters remains an active area of research to ensure the scalability and robustness of QKD systems.
  One such approach is based on so-called \emph{quantum repeaters}, which aim to enable secure long-distance quantum communication without the need for trusted intermediate nodes, leveraging quantum entanglement~\cite{quantum-repeaters}.
  However, as of today, quantum repeaters are still in the early experimental stage with significant challenges remaining, including the development of efficient quantum memories and the integration of all required components into a stable and scalable system.

  Another approach to reduce the necessary trust in the repeater nodes is to use multi-path QKD \cite{Liu_Che_Xie_Dong_2024,DBLP:conf/IEEEares/ValbusaLSL25}, which distributes quantum key material simultaneously over multiple independent paths in a network, increasing resilience against node or link failures and reducing the reliance on any single trusted repeater. 
  By combining the partial keys received via different paths--typically using information-theoretic techniques such as secret sharing \cite{Shamir79,secret-sharing-book}--the communicating parties can reconstruct a final key with improved security guarantees.
  Specifically, when using $k$-out-of-$t$ secret sharing for some $k<t$, one additionally obtains robustness guarantees: while the transferred data remains secure as long as at most $k-1$ paths contain a corrupted node, a joint key can be successfully established as long as at no more than $t-k$ paths refuse to relay data.

  However, although this approach mitigates the risk of a single compromised repeater, it has a fundamental limitation: 
  it requires the QKD network operator to be trusted to faithfully transmit the partial keys over pairwise disjoint paths. 
  This is because the security guarantees rely critically on the assumption that no individual node has access to more than one partial key. 
  Given that in the long term it can be expected that QKD networks will be operated by private entities such as telecommunications providers, and the stringent security requirements in scenarios like inter- and intra-governmental communication, relying solely on contractual instruments such as service level agreements (SLAs) may pose unacceptable risks.
  
  Therefore, technical and auditable measures are necessary to ensure that contractual agreements have been satisfied, e.g., verifying path disjointness, the certification level of deployed repeater nodes, or the host countries through which the communication path has been routed, without revealing sensitive business information about the topology of the underlying QKD network. 

\subsection{Our Contribution}
  In this paper, we propose a construction for proving that contractual agreements have been satisfied in a transmission over a QKD network. In particular, our construction is powerful enough to model the aforementioned use-cases.
  That is, it can be used to
  \begin{inumerate}
	\item ensure that the secret was transmitted over at least $k$ pairwise disjoint paths through the QKD network, and
	\item guarantee that only devices satisfying certain attributes (e.g., certification levels, manufacturers, etc.) were used during transmission, where
	\item the devices may have been certified by multiple authorities (differing, e.g., per country), 
	\item without leaking sensitive information about the network topology to the user of the network.
  \end{inumerate}
  
  More specifically, we present a modular extension applicable to arbitrary QKD networks, introducing only a small computational and communication overhead per repeater node and also at the receiver's end.
  Assuming that devices are certified before being deployed in practice--a necessary yet realistic assumption in high-security contexts which are subject to strict security controls--our construction can be used to audit that a QKD transmission has only passed through repeater nodes satisfying a previously defined policy;
  for concreteness, one might think of policies defining minimum certification levels for highly sensitive transactions, or trusted manufacturers or countries of origin of repeaters.
  At the same time, by using appropriate privacy-enhancing technologies (PETs), the receiving party only learns information about the number of repeaters on the path, which can anyhow be estimated based on the maximum distance of practical QKD networks.
  However, no further information about the topology is leaked.
  In particular, it remains hidden whether a specific node participated in different communication sessions, or to which other nodes it might be connected.
  
  To show the soundness of our approach, we provide a formal definitional framework capturing the intended security guarantees, and provide rigorous security proofs.
  Furthermore, we underpin the practicability of our solution by a detailed efficiency evaluation and benchmarks.
  
  \subsection{Technical Overview}
  In the following, we briefly sketch the core technical ideas underlying our construction, omitting specific details which will be discussed in the full construction. 
  While being conceptually easy to grasp, we want to stress that the benefit and added value of our contribution can significantly contribute to minimizing trust requirements in providers of QKD networks.
  
  Before deployment, each repeater node generates a local key pair $(\sk_\node,\pk_\node)$, and receives a certificate $\cred$ in the form of a digital signature on its attributes and its public key $\pk_\node$.
  
  When sending a message over a single path, the sender defines a policy $\predicate$ and chooses a unique session id $\sid$, which it transmits to the first node in the path.
  The node computes a non-interactive zero-knowledge proof of knowledge ($\NIZK$) $\pi_1$ that it owns a credential $\cred$ satisfying the policy $\predicate$, and binds it to the session id $\sid$.
  It then forwards $(\pi_1,\sid,\predicate)$ to the next node, which proceeds in a similar manner and sends $((\pi_1,\pi_2),\sid,\predicate)$ to the next node.
  Eventually, the receiver verifies the correctness of all received $\NIZK$s, thereby receiving guarantees that all nodes on the path satisfied the policy $\predicate$.
  
  Now, in a multi-path setting, each node on each path additionally computes a pseudonym $\nym$ for $\sid$, using a scope-exclusive pseudonym scheme.
  It then extends the $\NIZK$ to show that $\nym$ was derived from the secret key $\sk_\node$ underlying $\pk_\node$, which in turn was included in the credential $\cred$, without disclosing $\pk_\node$ to the verifier.
  At the end of the transmission, the receiver now checks that all received pseudonyms are pairwise distinct, thus receiving guarantees that no node was involved in more than one path.
  
  When aiming for trans-national QKD networks, nodes may be certified by different authorities (e.g., per country).
  As the verification of the $\NIZK$ requires access to the corresponding authority's public key $\pk_\issuer$, the receiver would learn information about the path.
  This can be overcome by leveraging the idea of issuer-hiding attribute-based credentials~\cite{CANS:BEKRS21}, which allow one to check that only accepted authorities issued the certificates while hiding the precise issuer.

\subsection{Related Work}
  While QKD itself is a highly active research area and a significant body of work aims at overcoming trust assumptions using, e.g., multi-path communication~\cite{Liu_Che_Xie_Dong_2024,DBLP:conf/IEEEares/ValbusaLSL25}, or at outsourcing computationally expensive parts of the post-processing~\cite{DBLP:journals/entropy/LorunserKPS23}, only limited work has focused on cryptographically auditing the behavior of network providers and nodes.
  As an example, Franzoi et al.~\cite{qcnc25} recently presented a protocol allowing one to identify the inconsistent link in cases where the integrity of the transferred secret was broken, i.e., that Alice and Bob did not obtain identical secret due to misbehavior of some repeater node. 
  Concurrent to our work, Cozzolino et al.~\cite{CozKreLor25} presented a protocol for showing the availability of QKD-links between two nodes in an inter-network scenario based on graph signatures, keeping the topology of the underlying networks secret.
  This is complementary to our effort in the sense that they prove availability of a (multi-path) connection, while we provide evidence that a used path satisfied certain constraints.
  
  Regarding mechanisms for \emph{path validation and path verification} for a next-generation Internet---e.g., to ensure that packages are routed via a superior (e.g., faster) network path as defined in the SLAs---have been researched, e.g, in~\cite{DBLP:conf/sigcomm/KimBJLHP14,DBLP:conf/conext/NaousWNMMS11,DBLP:journals/toit/SenguptaLBD20,DBLP:journals/tdsc/LiSLSMS25}.
  However, most importantly, these works all assume that the sender is aware of the entire path through the network, which is not applicable in our setting.
  Also, multi-path communication or complex policies $\predicate$ are not considered in these works, while the former could be easily achieved by letting the sender define disjoint transmission paths.
  At a more detailed level, there are trade-offs compared to our work in terms of privacy:
  for instance, \cite{DBLP:journals/toit/SenguptaLBD20} does not reveal the index (i.e., position on the path) of a node to that node, but discloses the length of the entire path to all nodes.
  In contrast, while revealing the index, we hide the path length to all but the last node in the communication.
  Notably, \cite{DBLP:journals/tdsc/LiSLSMS25} achieves a constant size proof, yet also relying on the knowledge of the entire path.
  While all the aforementioned work focuses on a single transmission path, Sengupta~\cite{DBLP:journals/cn/Sengupta22} lets senders define multiple valid paths, and the forwarding logic ensures that one of those paths was indeed followed.
  Note that this is different from our notion of multi-paths, which ensures that different parts of the data are sent over multiple, disjoint paths through the network.
  Summing up, while constituting a large body of work, existing path validation mechanisms are insufficient for our application, mainly because they assume that the sender is aware of the paths in the communication network, which is considered sensitive information in our scenario.

  Complementary to this, \emph{topology-hiding computation and communication} were introduced as privacy notions for running distributed protocols over an incomplete network while revealing essentially nothing about the underlying communication graph. 
  Moran, Orlov, and Richelson initiate topology-hiding computation and establish feasibility and limitations in early settings \cite{TCC:MorOrlRic15}. 
  Their work was extended to a general tool for higher-level protocols \cite{C:HMTZ16}, and extended to topology-hiding computation over all graph topologies \cite{C:AkaLaVMor17}.
  Later works refine the model toward stronger adversaries and more realistic timing assumptions, by pushing boundaries beyond semi-honest security \cite{TCC:LZMMMT18} and to networks with unknown delays \cite{PKC:LZMMMT20}. 
  Complementarily, Ball et al. investigate topology-hiding communication from the viewpoint of minimal setup and assumptions required to realize it \cite{TCC:BBCKMM20}.
  While all these works are closely related to our ambition of hiding the network structure while still carrying out nontrivial distributed tasks correctly (and in some cases even robustly), they pursue a different goal than ours: 
  they aim to realize generic topology-hiding communication or computation among protocol participants, and their guarantees are tied to the correctness and security of that interactive execution. 
  They do not provide a post-hoc, transferable audit guarantee for the receiver, namely an externally verifiable proof that the specific realized route satisfied policy constraints such as certified node properties and disjointness, which is what our approach is designed to deliver.
  
  In summary, our work can be viewed as bridging topology-hiding protocols with path validation/verification: 
  we adopt the topology-hiding objective to protect the operator's internal knowledge while importing the path-validation goal to provide externally verifiable evidence on the realized route's properties. 
  
\subsection{Outline}
  This document is structured as follows.
  In \cref{sec:preliminaries} we introduce the basic notation used in this paper, and recap the cryptographic building blocks used later on.
  Then, in \cref{sec:framework}, we formalize the syntax and security requirements for an auditable QKD protocol.
  We then introduce a generic construction achieving these properties in \cref{sec:generic_construction}, where we also provide rigorous security proofs.
  A specific instantiation is provided in \cref{sec:concrete_construction}, together with an efficiency analysis.
  Finally, we briefly conclude in \cref{sec:conclusion}.

\section{Preliminaries}\label{sec:preliminaries}
  Throughout the paper, we will denote the main security parameter by $\lambda$.
  We write $s\getsr S$ to denote that $s$ was sampled uniformly at random from a set $S$.
  Similarly, $s\getsr\alg{A}(x)$ denotes that $s$ is the output of a potentially randomized algorithm $\alg{A}$ on input $x$.
  For an interactive protocol between two parties $\alg{A}$ and $\alg{B}$ we write $(out_{\alg{A}};out_{\alg{B}})\getsr\langle\alg{A}(in_{\alg{A}}),\alg{B}(in_{\alg{B}})\rangle(in)$ to denote that $\alg{A}$ and $\alg{B}$ receive the corresponding outputs $out_{\alg{A}}$ and $out_{\alg{B}}$ upon executing the protocol on private inputs $in_{\alg{A}}$ and $in_{\alg{B}}$, as well as common input $in$.
  For an integer $n$, the set $\{1,\dots,n\}$ is denoted by $[n]$.
  A function $\negl$ is called negligible, if it vanishes faster than any inverse polynomial, i.e., for every integer $j$ there exists an integer $n_j$ such that $\negl(n)<n^{-j}$ for all $n>n_j$.
   We use $(\GG_1,\GG_2,\allowbreak\GG_T,e,p,G,\hat{G})$ to denote a bilinear group for asymmetric type-3 bilinear groups, where $p$ is a prime of bit length $\lambda$.

\subsection{Structure-Preserving Signatures}\label{sec:signatures}
  Digital signature schemes can be used to ensure the integrity, authenticity, and non-repudiation of digital messages.
  They consist of a quadruple of algorithms $(\ParGen,\KeyGen,\Sign,\Verify)$, where $\ParGen(1^\secpar)$ outputs system parameters $\pp$, $\KeyGen(\pp)$ outputs a key pair consisting of a secret signing key $\sk$ as well as a public verification key $\pk$, $\Sign(\pp,\sk,\msg)$ uses $\sk$ to derive signatures on messages $\msg$, and $\Verify(\pp,\pk,\sigma,\msg)¸$ checks the validity of a signature on a message relative to $\pk$.
  
  Digital signature schemes must satisfy EUF-CMA (existential unforgeability under chosen-message attacks), first formalized by Goldwasser et al.~\cite{GolMicRiv88}.
  This property requires that no PPT adversary, given only $\pk$, can compute a valid signature on a new message, even after having obtained valid signatures on arbitrarily many signatures of its own choice from a signing oracle.

  Structure-preserving signatures, first introduced by Abe et al.~\cite{C:AFGHO10}, are digital signatures where verification keys, messages, and signatures are elements of a bilinear group, and verification is done by checking a conjunction of pairing-product equations.

\subsection{Pseudonym Systems}\label{sec:pseudonyms}
  Pseudonym systems were first introduced by Chaum~\cite{DBLP:journals/cacm/Chaum85}, and later formalized in a series of work, e.g., \cite{SAC:LRSW99,SAC:CKLMNP15}.
  Such systems let users interact or authenticate using persistent aliases instead of their real identities, enhancing privacy and security.
  Specifically, in the case of scope-exclusive pseudonyms, pseudonyms are stable only within a given scope, while they cannot be linked across different scopes, even if derived from the same secret key.
  
  Pseudonym systems consist of three algorithms $(\ParGen,\KeyGen,\allowbreak\NymGen)$, where $\ParGen(1^\secpar)$ generates public parameters $\pp$, $\KeyGen(\pp)$ generates a user's secret key $\sk$, and the deterministic algorithm $\NymGen(\pp,\sk,\scope)$ computes a pseudonym $\nym$ for a given $\sk$ and a given $\scope$.

  Pseudonym systems need to be collision resistant, meaning that for each scope, any two users will have different pseudonyms.
  Furthermore, they have to be unlinkable, intuitively meaning that no PPT adversary can decide to which user a specific $\nym$ belongs, even after having received arbitrarily many pseudonyms of users on scopes chosen by the adversary.

\subsection{Zero-Knowledge Proofs of Knowledge}\label{sec:nizks}
  A zero-knowledge proof of knowledge (ZKP) is a protocol where a prover convinces a verifier that they know a piece of information without revealing the information about the secret beyond what is already revealed by the claim itself. 
  Unlike a basic zero-knowledge proof, which only demonstrates the existence of a fact, a ZKPoK proves that the prover possesses the knowledge of a secret.
  A ZKP is called non-interactive (NIZK), if the protocol consists of a single message sent from the prover to the verifier.
  
  A NIZK consists of a triple of algorithms $(\ParGen,\NIZK,\Verify)$, where $\ParGen$ generates the necessary system parameters, $\NIZK$ generates a non-interactive zero-knowledge proof of knowledge that the prover knows a secret witness $w$ such that $(x,w)\in\mathcal{R}$ for a binary relation $\mathcal{R}$ and a public value $x$, and $\Verify$ indicates whether to accept or to reject the proof.
  
  A NIZK has to be complete, meaning that an honest prover knowing a valid witness can always convince the verifier.
  Furthermore, zero-knowledge means that not knowing $w$ but knowing a simulation trapdoor, one can generate transcripts which are indistinguishable from real protocol transcripts.
  Finally, simulation sound extractability means that from any prover who can make the verifier accept, a valid witness can be extracted, even if the prover has seen simulated proofs for potentially false statements.
  For formal definitions, we refer, e.g., to Goldwasser et al.~\cite{STOC:GolMicRac85,AC:Groth06}.

  Slightly overloading notation, we will use Camenisch-Stadler notation~\cite{C:CamSta97} to denote proof goals.
  We write:
  $$
    \pi\getsr\NIZK\left[(\alpha,\beta,\gamma):Y=G^\alpha\cdot H^\beta ~\land~ Z=G^\alpha\cdot H^\gamma ~\land~ \gamma=\alpha\cdot\beta\right](\ctx)
  $$
  to prove knowledge of values $\alpha,\beta,\gamma$ such that the relation on the right hand side are satisfied;
  in particular, all values not in parentheses are assumed to be public.
  Furthermore, we often bind the proof to a context $\ctx$ in the sense of a signature of knowledge, for formal definition see, e.g.,~\cite{C:ChaLys06,AC:BCKLN14}.
  
\section{Auditable QKD Framework}\label{sec:framework}
In this section, we present a novel, privacy-preserving, and auditable protocol framework tailored to the requirements of large-scale QKD networks. Here we first formalize the notation and define the required security properties for auditable QKD protocol extensions.
\subsection{Syntax}
\begin{definition}[Auditable QKD]
	An auditing extension for a QKD protocol consists of the following algorithms:
\end{definition}
\begin{description}
	\item[$\pp\getsr\ParGen(1^\secpar)$.]
	  On input the security parameter $\lambda$ in unary representation, output the public parameters $\pp$.
	  These parameters are implicit input to all further algorithms, and might not be made explicit to ease presentation.
	\item[$(\sk_\issuer,\pk_\issuer)\getsr\KeyGen_\issuer(\pp)$.] 
	  This algorithm generates a secret key $\sk_\issuer$ as well as a public key $\pk_\issuer$ for an issuer.
	\item[$(\sk_\node,\pk_\node)\getsr\KeyGen_\node(\pp)$.]
	  This algorithm generates a secret key $\sk_\node$ as well as a public key $\pk_\node$ for a repeater node.
	\item[$(\cred;\bot)\getsr\langle \Register_\issuer(\sk_\issuer,\pk_\node),\Register_\node(\sk_\node,\pk_\issuer)\rangle(\vec{\attr})$.]
	  In this interactive protocol between an issuer and a repeater node, each party takes as input its own secret key and the other entity's public key, as well as previously agreed attributes $\vec{\attr}$ of a node.
	  At the end of the protocol, the node receives a certificate $\cred$ certifying its public key $\pk_\node$ and the attributes $\vec{\attr}$.
	\item[$(\bot;\{\bot\};n')\getsr\langle\Transfer_\Alice(),\{(\Transfer_\node(\sk_{\node_{ij}},\cred_{\node_{ij}},\vec{\attr}_{\node_{ij}})_{j=1}^{m_i})\}_{i=1}^n,\Transfer_\Bob()\rangle(\pk_\issuer,\predicate)$.]
	  \,\allowbreak This is an interactive protocol between a sender $\Alice$, a receiver $\Bob$, an a set of nodes along $n$ paths between $\Alice$ and $\Bob$, each containing $m_i$ nodes.
	  The parties take as inputs the issuer's key and the policy, and their respective secret inputs. 
	  At the end of the interaction, $\Bob$ receives as output an integer $n'$ indicating the number of disjoint paths used for the communication, or an error symbol to reject the transmission, while all the other parties do not receive any output.
\end{description}

\subsection{Security Overview and Adversarial Capacity}\label{sec:adversary-model}
In the following we briefly explain the considered adversary model and the intuition underlying our security framework presented in \cref{sec:security-defs}.

As discussed earlier, our ambition is to balance the needs of giving strong guarantees on the selected path to the customer, while at the same time protecting the confidentiality of the operator.

On the one hand, we define \emph{policy compliance}.
This property requires that a malicious network operator cannot forge proofs that a chosen path satisfied certain requirements.
That is, the operator's aim is to convince the communication partners $\Alice$ and $\Bob$ that a QKD-session was carried out over (one or more) communication links satisfying the agreed policy $\phi$ (specifying, e.g., a minimum certification level), while at least one of the involved repeaters did not fulfill the requirements.
In terms of the adversary's capabilities, we consider a malicious network operator controlling the communication network.
The operator has full control over the path used for a key exchange, and may install or remove arbitrary repeaters from the network.
However, we assume that key material of repeaters certified by an authority cannot be manipulated by the network operator -- a seemingly strong yet realistic assumption in the QKD context, which can be achieved, e.g., using secure hardware elements.
Furthermore, as is the case in practice, each repeater is assumed to know its physical neighbors in the network.

On the other hand, we introduce the notion of \emph{path-hiding}, which ensures that the operator's network topology remains hidden from the customer.
This is important as a provider's exact QKD network topology may reveal information about the location of trusted nodes, key-management infrastructure, and physical links, which exposes operational capabilities, redundancy, and potential single points of failure. 
This information could be exploited by competitors for strategic mapping and capacity inference, or by attackers to target the most valuable nodes and routes for disruption or surveillance.
We thus consider an overly strong adversary, who has full information about the network topology, and who is allowed to request an arbitrary (polynomial) number of communication sessions.
Despite this information, the adversary -- controlling the communicating parties $\Alice$ and $\Bob$ -- should not be able to learn anything about the selected path for a given communication session, except for the selected ``entry'' and ``exit'' nodes, which are necessarily revealed due to physical connectivity properties in QKD networks.
This will ensure that in particular no relevant information about a communication path is revealed to a less powerful adversary having less information about the overall network topology.

\subsection{Formal Security Definition}\label{sec:security-defs}

  Besides completeness, and as discussed in the previous section, we require two main properties from auditable QKD protocols, which we will formally define in the following.
  
  \paragraph{Path-Hiding.}
  This property ensures that the topology of the underlying QKD network is not revealed to the sender or receiver of a communication session.
  This is modeled in an indistinguishability game in Fig.\ref{fig:path_hiding}, where we assume that the adversary knows the precise graph of the QKD network.
  
  Let $\network$ be a description of the network graph, containing all nodes and edges. Now, in a first step, we generate all keys honestly, but let the adversary control the attributes of each node, which get certified honestly by the authority. We then give the adversary (controlling potentially multiple senders and receivers) access to two oracles:
  \begin{itemize}
  	\item   First, $\oracle$ allows the adversary to initiate arbitrary (multi-path) transmissions in the network using self-chosen policies over self-chosen paths.
  	\item   Second, the challenge oracle $\oracle_{LR}$, which may be called only once—enables the adversary to select two sets of paths. The oracle then initiates a transmission session using one of the two.
  \end{itemize}
  At the end of the experiment, the adversary needs to guess which paths were used for the transmission. In the oracles, we exclude trivial distinguishing features, like repeating session ids, paths of non-matching lengths, paths with loops, or inconsistent entry/exit nodes (note that $\Alice$ and $\Bob$ trivially know to/from whom they send/receive messages), or a policy that is not satisfied by (some of) the nodes on one of the paths.
  
  A scheme is now said to be path hiding, if no adversary has more than negligible advantage in winning the described experiment compared to random guessing. We formally define this game as follows:
  \begin{definition}[Path hiding]\label{def:path_hiding}
    An auditable QKD protocol is \emph{path hiding}, if for every PPT adversary $\adv$ there exists a negligible function $\negl$ such that for every network $\network$:
    $$
      \left|\Pr\left[\text{Exp}_{\text{path-hiding}}^{\adv}(\secpar)=1\right] - \frac{1}{2}\right|\leq\negl(\secpar),
    $$
    where the experiment is as defined in \cref{fig:path_hiding}.
  \end{definition}
\begin{figure}[th!]
	\centering\renewcommand{\arraystretch}{1.2}
	\begin{tabular}{|l|}
		\hline
		\rule{0pt}{2.5ex}\textbf{Experiment} $\text{Exp}_{\text{path-hiding}}^{\adv}(\secpar)$ \\ \hline
		$b\getsr\{0,1\}$\\
		$\setS\gets\emptyset$\\
		$\pp \getsr \ParGen(1^\secpar)$\\
		$(\sk_\node,\pk_\node)\getsr\KeyGen_\node(\pp)$ for all $\node\in\network$\\
		$(\sk_\issuer,\pk_\issuer)\getsr\KeyGen_\issuer(\pp)$\\
		$(\state,(\vec{\attr}_\node)_{\node\in\network})\gets \adv(\pp,\pk_\issuer,\{\pk_\node\}_{\node\in\network},\network)$\\
		$(\cred;\bot)\getsr\langle \Register_\issuer(\sk_\issuer,\pk_\node),\Register_\node(\sk_\node,\pk_\issuer)\rangle(\vec{\attr}_\node)$ for all $\node\in\network$\\
		$b'\getsr\adv^{\oracle(\cdot,\cdot,\cdot),\oracle_{LR}(\cdot,\cdot,\cdot,\cdot)}(\state)$\\
		\quad where $\oracle(\cdot,\cdot,\cdot)$ on input $(\sid,\predicate,\{(\node_{ij})_{j=1}^{m_i}\}_{i=1}^n)$:\\
		\quad\quad Aborts if $\sid\in\setS$, otherwise sets $\setS\gets\setS\cup\{\sid\}$, and engages in a transmission\\
		\quad\quad  protocol acting as $\{(\node_{ij})_{j=1}^{m_i}\}_{i=1}^n$ for the given $\sid$ and $\predicate$\\
		\quad where $\oracle_{LR}(\cdot,\cdot,\cdot,\cdot)$ on input $(\sid,\predicate,\{(\node_{ij}^0)_{j=1}^{m_i}\}_{i=1}^n,\{(\node_{ij}^1)_{j=1}^{m_i}\}_{i=1}^n)$:\\
		\quad\quad Aborts if it had been activated before, or if $\sid\in\setS$, or if the provided input\\
		\quad\quad paths contain loops, are not pairwise disjoint, or do not have corresponding\\
		\quad\quad lengths or entry/exit nodes, or any node does not satisfy $\predicate$. Otherwise, it sets \\
		\quad\quad $\setS\gets\setS\cup\{\sid\}$, and engages in a  transmission protocol acting as $\{(\node_{ij}^b)_{j=1}^{m_i}\}_{i=1}^n$ \\
		\quad\quad for the given $\sid$ and $\predicate$\\
		Return $b=b'$ \\ 
		\hline
	\end{tabular}
	\vspace{-1em}
	\caption{Path hiding}
	\label{fig:path_hiding}
\end{figure}
  Note that the above definition considers a single authority certifying all nodes in the network.
  In case, e.g., of cross-border communication, different authorities might be involved.
  If hiding the precise authority is important, the definition (as well as the construction presented below) can be adapted in a straightforward manner.

  \paragraph{Policy compliance.}
  Besides hiding the topology, our second main goal is to ensure that the nodes used for routing the message comply with the imposed policy and that the number of disjoint paths is correct.
  
  In our security definition, we again generate all keys honestly, and let the adversary choose the attributes of all devices in the network $\network$.
  Subsequently, the nodes are honestly registered.
  The adversary then gets full access to the nodes' certificates, allowing for more targeted attacks than by solely giving oracle access to the presentation of certificates.
  The adversary then has to define a session id $\sid$, a set of $n$ routes through the network, and a policy $\predicate$ of its choice.
  Based on these outputs, a transmission from $\Alice$ to $\Bob$ via the specified paths for the given session id and policy is initiated.
  The adversary now wins if $\Bob$ does not output an error symbol, and at least one of the following conditions is satisfied:
  \begin{inumerate}
  	\item At least one node in the transmission did not satisfy the policy $\predicate$, or
  	\item The different paths used for the transmission were not disjoint, or
  	\item The number of disjoint paths identified by $\Bob$ differs from the number of actual disjoint paths in the transmission.
  \end{inumerate}
  A scheme is now said to be policy compliant if no adversary has more than negligible advantage in winning the described experiment.
  
\begin{definition}[Policy compliance]\label{def:policy_compliance}
	An auditable QKD protocol is \emph{policy compliant}, if for every PPT adversary $\adv$ there exists a negligible function $\negl$ such that for every network $\network$:
	$$
	\Pr\left[\text{Exp}_{\text{policy-compliance}}^{\adv}(\secpar)=1\right] \leq\negl(\secpar),
	$$
	where the experiment is as defined in \cref{fig:policy_compliance}.
\end{definition}
\begin{figure}[th!]
	\centering \renewcommand{\arraystretch}{1.2}
	\begin{tabular}{|l|}
		\hline
		\rule{0pt}{2.5ex}\textbf{Experiment} $\text{Exp}_{\text{policy-compliance}}^{\adv}(\secpar)$ \\ \hline
		$\pp \getsr \ParGen(1^\secpar)$\\
		$(\sk_\node,\pk_\node)\getsr\KeyGen_\node(\pp)$ for all $\node\in\network$\\
		$(\sk_\issuer,\pk_\issuer)\getsr\KeyGen_\issuer(\pp)$\\
		$(\state,(\vec{\attr}_\node)_{\node\in\network})\gets \adv(\pp,\pk_\issuer,\{\pk_\node\}_{\node\in\network},\network)$\\
		$(\cred;\bot)\getsr\langle \Register_\issuer(\sk_\issuer,\pk_\node),\Register_\node(\sk_\node,\pk_\issuer)\rangle(\vec{\attr}_\node)$ for all $\node\in\network$\\
		$(\sid,\{(\node_{ij})_{j=1}^{m_i}\}_{i=1}^n,\predicate)\getsr\adv(\state,\{\cred_\node\}_{\node\in\network})$\\
		$(\bot;\{\bot\};n')\getsr\langle\Transfer_\Alice(),\{(\Transfer_\node(\sk_{\node_{ij}},\cred_{\node_{ij}},\vec{\attr}_{\node_{ij}})_{j=1}^{m_i})\}_{i=1}^n,\Transfer_\Bob()\rangle(\pk_\issuer,\predicate)$\\
		Return $1$ if $n'\ne\bot$ and:\\
		\quad $\phi(\vec{\attr}_{\node_{ij}})\ne 1$ for some $i\in[n]$ and $j\in[m_i]$, OR \hspace{1.30cm}\comment{policy violation}\\
		\quad $\{\node_{i_0j}\}_{j=1}^{m_{i_0}} \cap \{\node_{i_1j}\}_{j=1}^{m_{i_1}} \ne \emptyset$ for some $i_0\ne i_1$, OR\hspace{1.15cm}\comment{non-disjoint paths}\\
		\quad $n\ne n'$\hspace{6.95cm}\comment{wrong number of paths}\\
		Return $0$\\
		\hline
	\end{tabular}
	\vspace{-1em}
	\caption{Policy compliance}
	\label{fig:policy_compliance}
\end{figure}
\noindent Again, similar as above, an extension of the definition to multiple issuing authorities is straightforward.

\section{An Auditable QKD Protocol}\label{sec:generic_construction}
In the following, we present our auditable QKD protocol.
We first present a generic construction from well-known cryptographic building blocks, and show that it satisfies the requirements set out in \cref{sec:security-defs}.
We then give a concrete instantiation of the generic construction to show its practical usability.

\subsection{A Generic Construction}
  Our generic construction can be built from a signature scheme $\Sig$, a pseudonym scheme $\Nym$, and a non-interactive zero-knowledge proof system $\NIZK$, cf. \cref{sec:signatures,sec:pseudonyms,sec:nizks}.
  
  In a first step, the public parameters are generated by simply generating the parameters needed for the respective building blocks.
  Then, the issuer (authority) generates a key pair for a signature scheme, and each node samples the secret key for a pseudonym scheme.
  Furthermore, each node computes its public key as a pseudonym for a distinguished scope $\mathtt{setup}$.
  Now, for registration, the node first proves in zero-knowledge to the authority that it knows the secret key corresponding to its public key.
  The authority then uses its signing key to generate the certificate simply as a digital signature on the node's public key and the previously agreed attributes (note that validation of attribute values needs to happen out of band). 
  These onboarding steps are formally depicted in  \cref{constr:generic:setup}.

\begin{construction}[ht!]
\begin{minipage}{\textwidth}
\begin{description}\setlength{\itemsep}{5pt}
  \item[\underline{Parameter generation}.]
    The public parameters $\pp$ consist of the parameters $\pp_\Sig$ of the signature scheme $\Sig$, $\pp_\Nym$ of the pseudonym system $\Nym$, and $\pp_\NIZK$ of the non-interactive zero-knowledge proof system $\NIZK$.
    
  \item[\underline{Key generation}.]
    The issuer generates a key pair for a digital signature scheme as $(\sk_\issuer,\pk_\issuer)\getsr\Sig.\KeyGen(1^\secpar)$.

    Each node generates a key pair as $\sk_\node\getsr\Nym.\KeyGen(1^\secpar)$ and $\pk_\node\gets\Nym.\NymGen(\sk_\node,\mathtt{setup})$.

  \item[\underline{Registration}.]
    To register a node $\node$ with public key $\pk_\node$ and attributes $\vec{\attr}_\node$, the node first proves possession of the corresponding secret key:
    $$
      \pi \getsr \NIZK[(\sk_\node):\Nym.\NymGen(\sk_\node,\mathtt{setup})=\pk_\node](\did)\,,
    $$ 
    where $\did$ is a device registration id that is unique per registration process.
    The issuer then signs a credential $\cred\gets\Sig.\Sign((\pk_\node,\vec{\attr}_\node),\sk_\issuer)$ for the node.
\end{description}
\end{minipage}
\caption{Generic construction: Initialization and registration.}
\label{constr:generic:setup}
\end{construction}
  Now, when $\Alice$ wants to transmit a message to $\Bob$, they jointly agree on the number $n$ of disjoint paths to be used, and the policy $\predicate$ that has be satisfied by all nodes on the transmission paths.
  $\Alice$ then selects a unique session id $\sid$ and sends $(\sid,\predicate)$ to the entry node of each path.
  Note here that re-using a session id might leak information about the network topology, which however can be avoided, e.g., by ensuring freshness through inclusion of a timestamp validated by the entry nodes.
  Note further that $\Alice$ learns only the entry nodes of each communication path, which is unavoidable due to the physical connectivity in QKD systems;
  however, as modeled in \cref{def:path_hiding}, the remaining routing is oblivious to $\Alice$ and determined by the network operator.
  
  Upon receiving an (initially empty) set of zero-knowledge proofs and pseudo\-nyms $(\pi,\nym)$, as well as $\sid$ and $\predicate$, each node now computes its own pseudonym $\nym$ for the given session id.
  Subsequently, it computes a zero-knowledge proof of knowledge  $\pi_\node$ that $\nym_\node$ was indeed derived from the same secret key underlying the public key certified by the authority.
  Furthermore, the proof shows that the policy $\predicate$ is satisfied.
  As an additional safeguard against a service provider trying to introduce non-certified nodes, each node additionally computes a second zero-knowledge proof $\sigma$ showing that $\nym$ and its public key are indeed consistent, yet this time revealing its own public key.
  The node finally sends all previously received proofs $\pi$ and pseudonyms $\nym$, its own proof and pseudonyms $(\pi_\node,\nym_\node)$, as well as $\sigma$, $\sid$, and $\predicate$ to the next node on the path, who only accepts the inputs if $\sigma$ is valid. 
  See also \cref{fig:overview} for a visualization of this message flow.
  \begin{figure}[ht!]
  	\centering{\larger\larger\resizebox{\linewidth}{!}{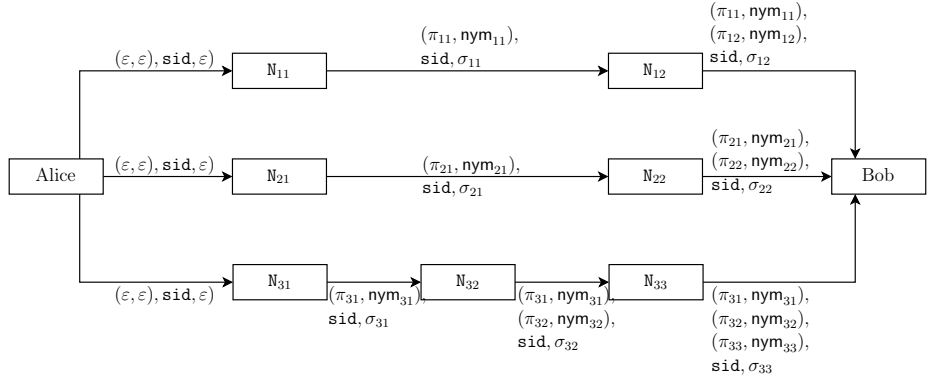}
  		\caption{Example message flow for three disjoint paths.}
  		\label{fig:overview}
  	}
  \end{figure}
  
  Eventually, $\Bob$ receives zero-knowledge proofs and pseudonyms from the exit nodes of all $n$ paths.
  $\Bob$ now checks whether all zero-knowledge proofs are valid, thereby receiving guarantees that all nodes involved in the transmission actually satisfied the policy $\predicate$ and that all received pseudonyms are authentic.
  Furthermore, $\Bob$ checks that all pseudonyms are pairwise distinct, thereby receiving guarantees that no node was involved twice in the transmission, and thus the paths were pairwise disjoint.
  
  Note that the additional proof $\sigma$ is not forwarded to $\Bob$, but always just sent to the subsequent node, who verifies its correctness and aborts otherwise.
  Given that in a QKD networks each node knows its neighbors (and thus also their public keys), this allows the node to check that the last $(\pi,\nym)$--sharing the same $\nym$ as $\sigma$--was indeed computed by the previous node in the path, as the public key underlying $\nym$ is revealed in $\sigma$.
  This prevents against service providers deploying non-certified nodes in the network, which would solely forward the received auditing data without appending their own information.
  
  A formal description of the protocol is now provided in \cref{constr:generic:routing}.

\begin{construction}[ht!]
\begin{minipage}{\textwidth}
\begin{description}\setlength{\itemsep}{5pt}
  \item[\underline{Sending}.]
    Alice defines the number $n$ of disjoint paths over which the secret has to be transmitted.
Furthermore, Alice agrees with Bob on a policy $\predicate$ to be satisfied by all nodes.
    She then selects a unique session id $\sid$, and sends  $((\varepsilon,\varepsilon),\sid,\predicate,\varepsilon)$ to the entry node of each path previously received from the service provider, i.e., to $\node_{i1}$ for all $i$.

  \item[\underline{Forwarding}.]
    A node $\node_{ij}$, upon receiving $(\msg=(((\pi_{i1},\nym_{i1}),\dots,(\pi_{i,j-1},\nym_{i,j-1})),\sid),\sigma)$, behaves as follows:
    \begin{itemize}
      \item
        If $j>1$, it fetches the public key $\pk_{i,j-1}$ of the sending node, and  verifies the validity of 
        the received 
        $\sigma$.
        The node then computes $\nym_{ij}\getsr\Nym.\NymGen(\sk_{ij},\sid)$ as well as the following two zero-knowledge proofs:
        \begin{align*}
          \pi_{ij} \gets \NIZK[&(\sk_{ij},\pk_{ij},\vec{\attr}_{ij},\cred_{ij}):\\
          & \Sig.\Verify((\pk_{ij},\vec{\attr}_{ij}),\cred_{ij},\pk_\issuer)=1 ~\land && \comment{the credential is valid}\\
          & \nym_{ij}=\Nym.\NymGen(\sk_{ij},\sid) ~\land  && \comment{$\nym_{ij}$ is well-formed and ...}\\
          & \pk_{ij}=\Nym.\NymGen(\sk_{ij},\mathtt{setup}) ~\land && \comment{... belongs to the right $\sk_{ij}$}\\
          & \predicate(\vec{\attr})=1](\msg) && \comment{the policy is satisfied}\\
          \sigma_{ij} \gets \NIZK[&(\sk_{ij}):\nym_{ij}=\Nym.\NymGen(\sk_{ij},\sid) ~\land \\
          &\pk_{ij}=\Nym.\NymGen(\sk_{ij},\mathtt{setup})](\msg)\,.
        \end{align*}
      \item 
        The node then transfers $((\pi,(\pi_{ij},\nym_{ij})),\sid,\predicate,\sigma_{ij})$ to the next node of each path, i.e., to $\node_{ij+1}$ or Bob in the case that $j=m_i$.
\end{itemize}

  \item[\underline{Receiving}.]
    Having received messages from $\node_{im_i}$ for all $i\in[n]$, Bob verifies all proofs and signatures.
    
    It furthermore validates that all received $\nym_{ij}$ are pairwise distinct. If all checks validate correctly, Bob outputs $n$, otherwise it aborts with $\bot$.
\end{description}
\end{minipage}
\caption{Generic construction: Message routing.}
\label{constr:generic:routing}
\end{construction}

\subsection{Security Analysis}
  We now formulate the main results of our work, and provide formal proofs that our generic construction indeed satisfies the properties defined in \cref{sec:security-defs}.
  
\begin{theorem}
	The auditable QKD protocol presented in \cref{constr:generic:setup,constr:generic:routing} is path-hiding in the sense of \cref{def:path_hiding}, if the deployed pseudonym system $\Nym$ is unlinkable and the proof system $\NIZK$ is simulation sound extractable.
\end{theorem}
\begin{proof}
	The proof follows standard techniques for privacy-enhancing protocols, and can be seen using the following series of games which are all computationally indistinguishable.
	
	\medskip\noindent\textbf{Game 0:} This is the original experiment as in \cref{def:path_hiding}.
	
	\medskip\noindent\textbf{Game 1:} We first modify the setup of the proof system $\NIZK$ by simulating the necessary parameters.
	
	This is indistinguishable by definition of simulation-sound extractability.

	\medskip\noindent\textbf{Game 2:} We now modify $\oracle_{LR}$ such that all zero-knowledge proofs $\pi_{ij}$ and $\sigma_{ij}$ generated by the different nodes $\node_{ij}^b$ for $i\in[n]$ and $j\in[m_i]$.
	
	By the zero-knowledge property of $\NIZK$ this is indistinguishable from the adversary's point of view.
	Note here that the adversary (controlling $\Alice$ and $\Bob$) only gets to see the $\pi_{ij}$, as well as the last $\sigma_{ij}$ per path, i.e., for $j=m_i$.
	Note further that in particular, after this game, all proofs are independent of the nodes' certificates $\cred_{ij}$, attributes $\vec{\attr}_{ij}$, and the secret keys $\sk_{ij}$.
	
	\medskip\noindent\textbf{Game 3:} We now change the computation of $\nym_{ij}=\Nym.\NymGen(\sk_{ij},\sid)$ in $\oracle_{LR}$ to $\nym_{ij}=\Nym.\NymGen(r_{ij},\sid)$ for fresh and independently sampled $r_{ij}\getsr\Nym.\KeyGen(1^\lambda)$.
	  
	Here, it is important to first note that $\sid$ is fresh, i.e., has not been used in any other transmission in the experiment, i.e., triggered through $\oracle$.
	The indistinguishability of this change then follows directly from the unlinkability of the pseudonym system, by which pseudonyms of the same and different users cannot be kept apart.
	
	\medskip
	
	The adversary's view in the last game is now fully independent of the choice of $b$, and thus the claim follows.
\qed
\end{proof}

\begin{theorem}
	The auditable QKD protocol presented in \cref{constr:generic:setup,constr:generic:routing} is policy-compliant in the sense of \cref{def:policy_compliance}, if the deployed signature scheme $\Sig$ is EUF-CMA secure, the pseudonym system $\Nym$ is collision resistant, and the proof system $\NIZK$ is simulation sound extractable.
\end{theorem}
\begin{proof}
	We prove the statement using a series of games.
	
	\medskip\noindent\textbf{Game 0:} This is the original experiment as in \cref{def:policy_compliance}.	
	
	\medskip\noindent\textbf{Game 1:} In this game, we modify the setup of the proof system $\NIZK$ to simulate the necessary parameters.
	
	This change is indistinguishable by definition of simulation-sound extractability.
	
	\medskip\noindent\textbf{Game 2:} In this game, we use the trapdoor from the previous game to extract the witnesses of all proofs $\pi_{ij}$ and $\sigma_{ij}$ received by $\Bob$.
	That is, we obtain values $\cred_{ij}$, $\vec{\attr}_{ij}, \sk_{ij}, \pk_{ij}$ satisfying the relations stated in $\pi_{ij}$ and $\sigma_{ij}$ as described in \cref{constr:generic:routing}, respectively.
	If the extraction of valid witnesses fails for either of the proofs, the experiment aborts and the adversary wins.
	
	By the simulation sound extractability of the proof system, the difference between these two games is negligible.
	
	\medskip\noindent\textbf{Game 3:} In this game, we additionally check that $\vec{\attr}_{ij}$ and $\pk_{ij}$ have indeed been signed by the issuing authority as part of valid registration processes.
	If this is not the case, we abort.
	
	By the EUF-CMA security of the signature scheme, an abort only happens with negligible probability.

	\medskip\noindent\textbf{Game 4:} In this game, we abort if two different nodes generated identical pseudonyms for the given $\sid$.
	
	By the collision resistance of the pseudonym system, this can only happen with negligible probability.
	
	\medskip
	
	Overall, we are now at a point where we extracted from each node $\node_{ij}$ a set of attributes $\vec{\attr}_{ij}$ satisfying the required policy $\predicate$, and previously certified by the issuing authority together with a public key $\pk_{ij}$ which was also used to compute $\nym_{ij}$.
	Furthermore, the pseudonyms were all honestly generated and pairwise different, proving the disjointness of the different paths.
	Finally, the number $n'$ of paths obtained by $\Bob$ corresponds to the claimed number $n$ by inspection of the construction (note here that in the definition we assume that the protocol is executed by honest nodes).

    The claim of the statement now follows immediately.
\qed
\end{proof}

The proof strategy shows that the security proof remains valid even if the adversary controls all secret key material of all nodes before the protocol starts. In particular, the key pair $(\sk_\node,\pk_\node)$ need not be honestly generated.
However, this stronger model does not reflect realistic scenarios, as nodes are typically subject to thorough auditing.
We therefore chose not to adopt this stronger model in \cref{def:policy_compliance}, to avoid ruling out constructions that align with real-world requirements.

\subsection{Multiple issuers}
  In realistic deployments, it might happen that QKD links involving multiple network operators have to be established, e.g., when communicating across country borders.
  In this case, as already briefly mentioned in \cref{sec:security-defs}, a single issuing authority might not be realistic anymore, but rather one issuer, e.g., per EU member state might exist.
  Using the framework and construction presented so far, it would now be leaked to the receiver who certified the different repeaters, as the issuer's public key $\pk_\issuer$ is used when verifying the zero-knowledge proofs, thereby disclosing some minimum information about the path being used.
  
  We deliberately excluded this scenario from our definitional framework and main construction to maintain clarity and comprehensibility.
  Yet, this limitation can be easily addressed if needed, using the concept of issuer-hiding attribute-based credentials, first introduced by Bobolz et al.~\cite{CANS:BEKRS21}, and later picked up also by, e.g., \cite{CCS:MBGLS23,PKC:ConLafPer22,PoPETS:BFGP22}.
  Such systems allow one to hide the precise issuer of a credential, while still guaranteeing that it was among a set of accepted issuers.
  
  Following the approach of \cite{CANS:BEKRS21}, this could be achieved by modifying the computation of $\pi_{ij}$ as follows.
  Instead of proving knowledge of $\pk_{ij}, \vec{\attr}_{ij}, \cred_{ij}$ such that 
  $$
  \Sig.\Verify((\pk_{ij},\vec{\attr}_{ij}),\cred_{ij},\pk_\issuer)=1\,,
  $$ 
   one would prove knowledge of additional values $\pk_\issuer$ and $\tau_{\pk_\issuer}$ such that 
  $$
    \Sig.\Verify((\pk_{ij},\vec{\attr}_{ij}),\cred_{ij},\pk_\issuer)=1\,\land\,\Sig'.\Verify(\pk_\issuer,\tau_{\pk_\issuer},\pk')=1\,.
  $$ 
  Here, $\Sigma'$ is an appropriate signature scheme, $\tau_{\pk_\issuer}$ is a signature on the public key of the issuer under some $\pk'$ trusted by $\Alice$ and $\Bob$.
  Depending on the scenario, $\pk'$ can be a fresh key for which all acceptable issuers (e.g., those of EU27 member states) are signed for a single transmission.
  Alternatively, the signatures $\tau_{\pk_\issuer}$ could also be computed and published by a trusted entity, and used across multiple sessions.
  For further details on this approach, we refer to the original work on issuer-hiding credentials.

\section{Concrete Instantiation}\label{sec:concrete_construction}
  In the following we now present a specific instantiation of our generic construction, in order to prove the practicability of our approach.
  
  \subsection{Building Blocks}
    We first detail the instantiations of the building blocks introduced in \cref{sec:signatures,sec:pseudonyms,sec:nizks}.

    \paragraph{Structure-Preserving Signatures.}
      In the following, we present a version of Groth's scheme~\cite{AC:Groth15} tailored to our needs, as we only need to sign a single group element.
      Similar to previous work~\cite{CCS:CamDriDub17,CANS:BEKRS21} we thereby consider a version $\Groth_1$ signing elements of $\GG_1$.
      More specifically, as our construction in the following only needs to sign a single element in $\GG_1$, we also restrict the presentation in the following to such a case.

\begin{itemize}
	\item 
	$\Groth_1.\ParGen(1^\secpar)$ generates public parameters of a bilinear group $(\GG_1,\GG_2,\allowbreak\GG_T,e,p,G,\hat{G})$ of prime order $p$, where $G\in\GG_1$ and $\hat{G}\in\GG_2$ are generators.
	It additionally outputs an element $Y\getsr\GG_1$.
	\item 
	$\Groth_1.\KeyGen(\pp)$ samples $\sk\gets\ZZ_p^*$ and sets $\pk\gets\hat{G}^\sk$.
	\item 
	$\Groth_1.\Sign(\pp,\sk,\msg)$ samples $r\gets\ZZ_p^*$ and computes a signature $\sigma=(\hat{R},S,T)=(\hat{G}^r,(Y\cdot G^\sk)^{1/r},(Y^\sk\cdot\msg)^{1/r})$.
	\item 
	$\Groth_1.\Rand(\pp,\sigma)$ rerandomizes a valid signature $\sigma$ by sampling $r'\gets\ZZ_p^*$ and outputting a randomized signature $\sigma'=(\hat{R}',S',T')=(\hat{R}^{r'},S^{1/r'},T^{1/r'})$ on the same message.
	\item 
	$\Groth_1.\Verify(\pp,\pk,\sigma,\msg)$ outputs $1$ if and only if it holds that $e(S,\hat{R})=e(Y,\hat{G})\cdot e(G,\pk)$ and $e(T,\hat{R})=e(Y,\pk)\cdot e(\msg,\hat{G})$.      
\end{itemize}

      As shown by Groth~\cite{AC:Groth15}, the scheme satisfies EUF-CMA in the generic group model.
      Furthermore, it is easy to see that outputs of $\Groth_1.\Rand$ follow the same distribution as fresh signatures on the same message.

      In the following we will not make $\pp$ explicit as input to the algorithms.
      Furthermore, a dual scheme signing elements in $\GG_2$ can easily be obtained by switching the roles of $\GG_1$ and $\GG_2$.

    \paragraph{Pseudonyms.}
      The following description follows that of Camenisch et al.~\cite{SAC:CKLMNP15}.
\begin{itemize}
	\item 
	$\SENym.\ParGen(1^\secpar)$ outputs a parameters $\pp$ specifying a group $\GG$ of prime order $p$ together with a generator $G$ of $\GG$.
	Furthermore, $\hash:\zo^*\to\GG$ specifies a cryptographic hash function.
	\item 
	$\SENym.\KeyGen(\pp)$ samples a user secret key as $\sk\getsr\ZZ_p^*$.
	\item 
	$\SENym.\NymGen(\pp,\sk,\scope)$ computes a scope-exclusive pseudonym as $\nym\gets\hash(\scope)^\sk$.
\end{itemize}
      Under the DDH assumption and in the random oracle model, the scheme can be shown to be unlinkable and collision resistant.

      In the following we will not make $\pp$ explicit as input to the algorithms.

    \paragraph{Zero-Knowledge Proofs.}
      Depending on the specific proof goal, different instantiations can be used.
      All our proof goals correspond to proving the secret preimage of a homomorphism, such that the Schnorr protocol~\cite{C:Schnorr89} can be used.
      The protocol can be made turned into a non-interactive, simulation sound extractable variant using the Fiat-Shamir heuristic~\cite{C:FiaSha86,AC:BerPerWar12}.

      Alternatively, also Groth-Sahai proofs~\cite{EC:GroSah08}, which also satisfy simulation-sound extractability, could be deployed for all proof goals used in our construction.

  \subsection{Instantiating the Generic Construction}
    Now, putting all together, we obtain the construction presented in \cref{constr:inst:setup,constr:inst:routing}.
    Note that in the forwarding phase in \cref{constr:inst:routing}, the efficiency of the zero-knowledge proof was slightly increased by not proving explicitly that $\pk=\Nym.\NymGen(\sk_{ij},\mathtt{setup})$.
    Rather, this is proven implicitly by not proving that $\pk_{ij}$ is contained in the credential $\cred$ but that the node knows the discrete logarithm of the signed message, i.e., that $\hash(\texttt{setup})^{\sk_{ij}}$ is signed.
    Importantly, this syntactical change does not modify the proof goal, such that the security proofs of the generic construction still holds for the instantiation.
    
\begin{construction}[th!]
	\begin{minipage}{\textwidth}
		\begin{description}\setlength{\itemsep}{5pt}
			\item[\underline{Parameter generation}.]
              Generates public parameters of a bilinear group $(\GG_1,\GG_2,\allowbreak\GG_T,e,p,G,\hat{G})$ of prime order $p$, where $G\in\GG_1$ and $\hat{G}\in\GG_2$ are generators.
              Sample $Y\getsr\GG_1$.
              Define a cryptographically secure hash function $\hash:\{0,1\}^*\to\GG_1$.
              Finally, sample $H_i\getsr\GG_1$ for $i\in[\ell]$, where $\ell$ is the maximum number of supported attributes.
              
              Output $\pp\gets(\GG_1,\GG_2,\GG_T,e,p,G,\hat{G},Y,\hash,(H_i)_{i=1}^\ell)$.

			\item[\underline{Key generation}.]
			The issuer computes a Groth key pair as $\sk_\issuer\getsr\ZZ_p^*$ and $\pk_\issuer=\hat{G}^{\sk_\issuer}$.
			
			Each node samples $\sk_\node\getsr\ZZ_p^*$ and sets $\pk_\node=\hash(\mathtt{setup})^{\sk_\node}$.
			
			\item[\underline{Registration}.]
			To register, a node $\node$ with attributes $\vec{\attr}_\node$ computes:
			$$
			\pi \getsr \NIZK[(\sk_\node):\pk_\node=\hash(\mathtt{setup})^{\sk_\node}](\did)\,,
			$$ 
			where $\did$ is a nonce chosen by the issuer.
			The issuer then computes $\msg$ as a Pedersen hash to $\pk_\node$ and $\attr_\node$ as $\msg=\pk_\node\cdot\prod_{i=1}^\ell H_i^{\attr_i}$. 
			It then computes the credential as a Groth signature, i.e., it samples $r\gets\ZZ_p^*$ and sets $\cred=(\hat{R},S,T)=(\hat{G}^r,(Y\cdot G^{\sk_\issuer})^{1/r},(Y^{\sk_\issuer}\cdot \msg)^{1/r})$.
		\end{description}
	\end{minipage}
	\caption{Concrete instantiation: Initialization and registration.}
	\label{constr:inst:setup}
\end{construction}

    Given that the resulting proof goal still realizes the proof goal in the generic construction, it is important to note that the security analysis of the generic construction still holds true for the instantiation.

\begin{construction}[ht!]
	\begin{minipage}{\textwidth}
		\begin{description}\setlength{\itemsep}{5pt}
			\item[\underline{Sending}.]
			This step is identical to the generic construction.
			\item[\underline{Forwarding}.]
			A node $\node_{ij}$, upon receiving $(\msg=(((\pi_{i1},\nym_{i1}),\dots,(\pi_{i,j-1},\nym_{i,j-1})),\sid),\sigma)$, behaves as follows:
			\begin{itemize}
				\item
				If $j>1$, it fetches the public key $\pk_{i,j-1}$ of the sending node, and  verifies the validity of the received 
				$\sigma$.
				\item
				The node re-randomizes $\cred=(\hat{R},S,T)$ by sampling $r'\gets\ZZ_p^*$, obtaining $(\hat{R}',S',T')=(\hat{R}^{r'},S^{1/r'},T^{1/r'})$.
				It blinds the result by sampling $\alpha,\beta\getsr\ZZ_p^*$ and setting $(\hat{R}'',S'',T'')=(\hat{R}',S'^{1/\alpha},T'^{1/\beta})$.
				The node computes $\nym_{ij}\getsr\allowbreak\hash(\sid)^{\sk_{ij}}$ as well as the following two Schnorr-like zero-knowledge proofs:
				\begin{align*}
					\pi_{ij}' \gets \NIZK[&(\sk_{ij},\vec{\attr}_{ij},\alpha,\beta):\\
					& 
					e(S'',\hat{R}'')^\alpha = e(Y,\hat{G})\cdot e(G,\pk_\issuer) ~\land \\
					&e(T'',\hat{R}'')^\beta = e(Y,\pk_\issuer)\cdot e\left(\hash(\mathtt{setup})^{\sk_{ij}}\prod_{i=1}^\ell H_i^{\attr_i},\hat{G}\right)
					\\
					& \nym_{ij}=\hash(\sid)^{\sk_{ij}} ~\land \\
					& \predicate(\vec{\attr})=1](\msg) \\
					\sigma_{ij}' \gets \NIZK[&(\sk_{ij}):\nym_{ij}=\hash(\sid)^{\sk_{ij}} ~\land ~
					\pk_{ij}=\hash(\mathtt{setup})^{\sk_{ij}}](\msg)\,.
				\end{align*}
				\item 
				The node finally sets $\pi_{ij}=(\pi_{ij}',\hat{R}'',S'',T'')$ and transfers $((\pi,(\pi_{ij},\nym_{ij})),\sid,\sigma_{ij})$ to the next node of each path, i.e., to $\node_{ij+1}$ or Bob in the case that $j=m_i$.
			\end{itemize}
			
			\item[\underline{Receiving}.]
			This step is identical to the generic construction.
		\end{description}
	\end{minipage}
	\caption{Concrete instantiation: Message routing.}
	\label{constr:inst:routing}
\end{construction}

  \subsection{Efficiency Analysis}
  In the following we estimate the actual complexity of our construction.
  As the key generation and registration protocols are only executed once per node, they are omitted in the analysis.
  
  We thus first count the number of predominant operations (i.e., full-length exponentiations as well as pairing evaluations, yet omitting simple group operations or operations in $\ZZ_p$ as well as hash evaluations) for each node $\node_{ij}$ in the transmission as well as for the final receiver $\Bob$, cf. \cref{tab:computation_costs}, where for clarity we break down the costs for each step in the computation.
  Note here that the receiver only has to verify one $\sigma_{ij}'$, but all proofs $\pi_{ij}'$ generated during the transmission, and thus the overall costs are linear in the number $n$ of involved nodes.

\begin{table}[h!]
	\centering
    \begin{tabular}{p{0.4cm} p{3.4cm} C{1.9cm} C{1.9cm} C{1.9cm} C{1.9cm}}    
  	  \toprule
	  & & $\GG_1$ & $\GG_2$ & $\GG_T$ & $e(\cdot,\cdot)$  \\
	  \midrule
    \multirow{6}{*}{\rotatebox[origin=c]{90}{Repeater node}}
      & Verifying $\sigma_{ij}'$       & $4$ & -- & -- & --\\ 
  	  & Computation of $\nym_{ij}$     & $1$  & -- & -- & --\\ 
	  & Deriving $(\hat{R}'',S'',T'')$ & $4$  & $1$  & -- & --\\ 
	  & Computing $\pi_{ij}'$          & $\ell+2$ & -- & $2$  & $4$\\ 
	  & Computing $\sigma_{ij}'$       & $2$ & -- & -- & --\\ 
	  \cmidrule{2-6}
      & Total & $\ell+13$ & $1$ & $2$ & $4$\\
	  \midrule
    \multirow{3}{*}{\rotatebox[origin=c]{90}{Receiver}}
& Verifying final $\sigma_{ij}'$       & $4$ & -- & -- & --\\ 
& Verifying all $\pi_{ij}'$       & $(\ell+3)\cdot n$ & -- & $4\cdot n$ & $4\cdot n$\\ 
	  \cmidrule{2-6}
& Total & $(\ell+3)\cdot n + 4$ & -- & $4\cdot n$ & $4\cdot n$\\
\bottomrule
    \end{tabular}
\vspace*{2pt}
\caption{Computational costs per node}
\label{tab:computation_costs}
\end{table}

Besides this theoretical analysis, we also provide actual benchmarks from a prototypical implementation of our application for different numbers of nodes ($n=10,\dots,100$), attributes ($\ell=10,20$ assuming that half of them are disclosed), and for single and multi-path settings (in the single-path setting, no pseudonyms are required).
The benchmarks were carried out on an Intel Core 2 Duo E8400 CPU with $3.00$GHz and $4$GB of RAM, running Ubuntu 24.04.2 LTS and Python 3.12.3, using the \texttt{mcl 3.00}\footnote{\url{https://github.com/herumi/mcl}} library for pairing-based crypto.
To compensate for the relatively weak hardware of the setup, the figure also includes estimates for an AWS Linux 2 8-core Intel Xeon Platinum 8259CL CPU running at $2.50$GHz using $32$GB of RAM, obtained using the \emph{zkalc estimator}\footnote{\url{https://zka.lc}} using the \texttt{gnark-crypto} implementation.
As a pairing friendly curve, we opted for the BLS12-381 curve, offering about 120bit of security.

The benchmark results are now depicted in \cref{fig:benchmarks}.%
\footnote{The source code of the benchmark implementation can be accessed at \url{https://anonymous.4open.science/r/auditable-qkd-55EF/}.}
The figure only presents the computational costs of the receiver $\Bob$.
The computational costs of all intermediate nodes were well below $20$\,ms and thus in the area of network latency, such that they can be ignored in practice.
As could be expected from \cref{tab:computation_costs}, the receiver's runtime is linear in the number $n$ of nodes, while the number $\ell$ or the costs of pseudonyms only play a secondary role.

\begin{figure}[h!]
	\centering
	 \includegraphics[width=0.9\textwidth]{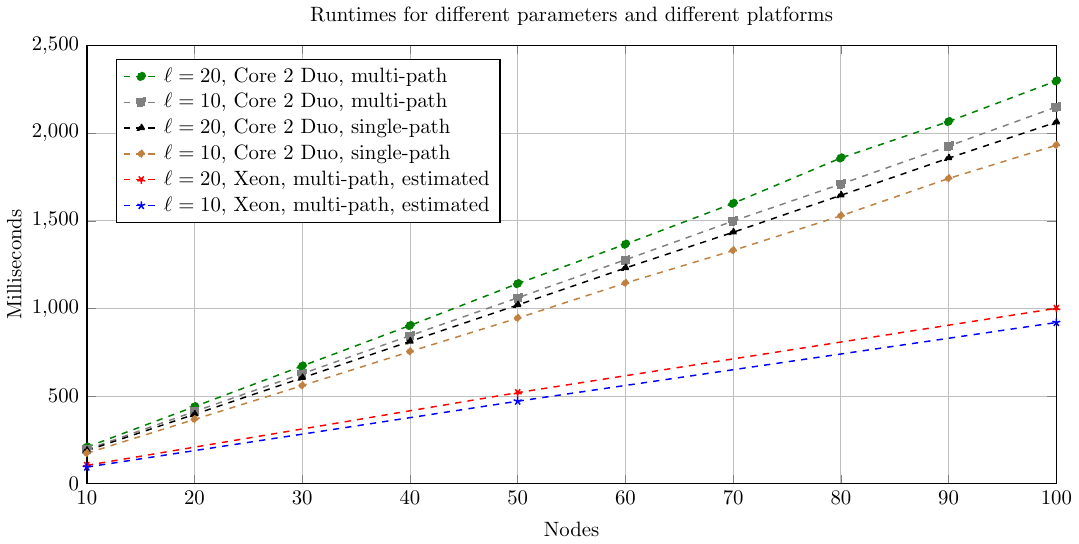} 
	\vspace*{2pt}
	\caption{Runtime evaluation for receiver $\Bob$.}
	\label{fig:benchmarks}
\end{figure}

As can be seen, even for a very high number of $100$ repeater nodes (corresponding, e.g., to $3$ disjoint paths over more than $3'300$\,km considering current distance limitations, thus exceeding the maximum distance between any two EU capitals in mainland Europe, i.e., Porto/Portugal and Helsiki/Finland), the computational overhead is in the area of less than $2.5$\,s and less than $1$\,s on more powerful hardware, both without applying any optimizations like batching or pre-processing.
We consider this overhead acceptable for typical quantum-key distribution deployments, where key material is often generated ahead of use and buffered rather than produced strictly on-demand. 
In practice, QKD keys are commonly generated in batches during periods of link availability (for example, during connectivity windows in satellite or hybrid networks) and stored in a key store or HSM for subsequent consumption. 
As a result, the introduced verification delay generally impacts key availability only at session boundaries and does not translate into added data-path latency or reduced usability for most applications.

\medskip

Finally, \cref{tab:communication_costs} shows the overall bandwidth overhead for the receiver $\Bob$, i.e., the overall size of the received data.
We omit the bandwidth requirements of the repeater nodes, as the size of transferred data grows linearly from the first node to the receiver, such that all intermediate nodes have less bandwidth requirements than $\Bob$.
Note that in the table, following the notation from the protocols, $\pi_{ij}$ also includes the re-randomized signatures $(\hat{R}'',S'',T'')\in\GG_2\times\GG_1^2$ per node.
In the table, we further assume that the policy $\predicate$ requires to disclose $d$ of the $\ell$ attributes, while the other $\ell-d$ attributes remain undisclosed.

\begin{table}[h!]
	\centering
	\begin{tabular}{p{3.5cm} C{2.0cm} C{2.0cm} C{2.0cm} C{2.0cm}}    
		\toprule
		& $\GG_1$ & $\GG_2$ & $\GG_T$ & $\ZZ_p$  \\
		\midrule
		Final $\sigma_{ij}$    & --  & -- & -- & 2\\ 
		Proofs $\pi_{ij}$      & $2\cdot n$  & $n$ & -- & $5\cdot n$\\ 
		Pseudonyms $\nym_{ij}$ & $n$ & -- & -- & -- \\
		\midrule
		Total & $3\cdot n$ & $n$ & -- & $(4+d)\cdot n+2$\\
		\bottomrule
	\end{tabular}
	\vspace*{2pt}
	\caption{Bandwidth overhead for receiver $\Bob$}
	\label{tab:communication_costs}
\end{table}

When estimating the actual bandwidth requirements, we assume a size of $48$\,B for elements in $\GG_1$, $96$\,B for elements in $\GG_2$, and $32$\,B for elements in $\ZZ_p$, which are appropriate values for BLS12-381.
As before, we further assume that half of the attributes are disclosed and half of them remain private.
This then results in an overall communication size for $\Bob$ of about $67$\,kB for $100$ repeaters and $20$ attributes.
Given the substantial classical communication already required by standard QKD post-processing, this additional overhead can be considered negligible in practice. 
For example, in the original (unbiased) BB84 protocol~\cite{C:BenBra84}, basis reconciliation alone requires exchanging basis information for each detection event, which corresponds to a few bits per detection/sifted bit, even before accounting for parameter estimation, error correction, and authentication.

Considering the above estimates, we thus consider the proposed protocol practical for many application scenarios.


\section{Conclusions and Open Challenges}\label{sec:conclusion}
In this work, we presented a protocol to validate the paths over which a QKD transmission has been carried out.
To do so, we formally defined the framework including the desired security properties, and provided a generic construction which we formally proved secure.
Our concrete instantiation shows the practicality of approach, also for long-distance QKD transmissions with a high number of intermediate repeater nodes.


While our main focus is on QKD networks, the general idea of topology-hiding path validation may be useful also in other settings where an operator wants to keep internal routing and infrastructure confidential, but endpoints still require verifiable evidence that traffic complied with externally specified constraints. 
This is complementary to traditional approaches such as source-routing, where the sender fixes (and therefore learns) the end-to-end path; 
in many operational networks routing is hop-by-hop and policy-driven, so requiring the sender to specify, or even know, the route is often unrealistic.
	
Concretely, one could imagine applications in managed, compliance-driven networks such as carrier backbones (e.g., proving traversal only through ``hardened'' equipment from an approved vendor set) or critical-infrastructure interconnects (e.g., power-grid, rail, or emergency-services backhaul, where regulators or customers may require evidence that traffic avoided uncertified nodes or certain geographic regions). 
Related scenarios include confidential content delivery networks (CDNs) (e.g., for proving traffic was handled only by attested gateways or enclaves), and high-assurance multi-path delivery for resilience (e.g., proving that replicated deliveries used disjoint provider resources to reduce correlated failure risk). 
We emphasize that such uses make sense primarily when hop counts are moderate or sessions are long-lived (so linear proof growth can be amortized), and when deployments already support per-hop cryptographic identities or attestations;
within those regimes, our approach offers a practical way to couple routing privacy with receiver-verifiable policy compliance.

Finally, interesting open challenges include achieving index-hiding, where a node does not learn its position in the path.
Furthermore, overcoming the current linear growth of the audit proofs with the path length would be interesting to broaden the applicability of our approach. 
One possible approach is to use recursive SNARKs~\cite{C:KotSet24}, where each node proves in constant-size form that it verified the previous proof and correctly applied the prescribed local extension, yielding proof size and receiver verification cost that are independent of the number of hops. 
A key challenge is to retain practical efficiency under standard deployment constraints and, in particular, to support succinct proofs of multi-path properties such as node-disjointness without reintroducing linear overheads.

\subsubsection*{Acknowledgments.}
This work has received funding from the European Union's Horizon Europe research and innovation program under No. 101114043 (\textsc{Qsnp}), from the Digital European Program under No. 101091588 (\textsc{Quarter}) and No. 101091642 (\textsc{Qci-Cat}), and the National Foundation for Research, Technology and Development.
Views and opinions expressed are however those of the authors only and do not necessarily reflect those of the funding institutions. 
Neither the funding institutions nor the granting authorities can be held responsible for them.

The authors would also like to thank the anonymous reviewers for detailed comments and suggestions on how to improve the presentation and positioning of the result.

  \bibliographystyle{splncs04}


\bibliography{additional,cryptobib/abbrev0,cryptobib/crypto}
\end{document}